%% file: attractors-with-reduction.tex
\documentclass{article}

\usepackage[margin=1.1in]{geometry}
\usepackage{authblk}
\usepackage{amsmath}
\usepackage{amssymb}
\usepackage{amsthm}
\usepackage{tikz}
\usepackage{tikz-cd}
\usetikzlibrary{hobby}
\usetikzlibrary{arrows,shapes,shapes.misc,positioning,calc}
\usepackage{hyperref}
\usepackage[capitalise]{cleveref}
\usepackage{makecell}
\usepackage{caption}
\newcommand{\A}{\mathcal{A}}
\newcommand{\AD}{\Gamma}

\newcommand{\tf}{\tilde{f}}
\renewcommand{\S}{\mathcal{S}}

\newtheorem{proposition}{Proposition}[section]
\newtheorem{theorem}[proposition]{Theorem}

\theoremstyle{definition}

\newtheorem{example}[proposition]{Example}

\begin{document}
\title{Attractor identification in asynchronous Boolean dynamics\\with network reduction}
\author[1]{Elisa Tonello}
\author[2]{Loïc Paulev\'{e}}
\affil[1]{Freie Universität Berlin, Germany}
\affil[2]{Univ. Bordeaux, CNRS, Bordeaux INP, LaBRI, UMR 5800, F-33400 Talence, France}

\maketitle

\begin{abstract}
  \input{abstract}
\end{abstract}

\input{content}

\subsubsection*{Acknowledgements}
ET was supported by the Deutsche Forschungsgemeinschaft (DFG) under Germany's Excellence Strategy – The Berlin Mathematics Research Center MATH+ (EXC-2046/1, project ID 390685689).
LP was supported by the French Agence Nationale pour la Recherche (ANR) in the scope of the
project ``BNeDiction'' (grant number ANR-20-CE45-0001).
Experiments presented in this paper were carried out using the PlaFRIM experimental testbed, supported by Inria, CNRS (LABRI and IMB), Université de Bordeaux, Bordeaux INP and Conseil Régional d’Aquitaine (see https://www.plafrim.fr).

\bibliography{biblio}
\bibliographystyle{plain}

\end{document}

%% file: abstract.tex
Identification of attractors, that is, stable states and sustained oscillations, is an important step in the analysis of Boolean models and exploration of potential variants. We describe an approach to the search for asynchronous cyclic attractors of Boolean networks that exploits, in a novel way, the established technique of elimination of components. Computation of attractors of simplified networks allows the identification of a limited number of candidate attractor states, which are then screened with techniques of reachability analysis combined with trap space computation. An implementation that brings together recently developed Boolean network analysis tools, tested on biological models and random benchmark networks, shows the potential to significantly reduce running times.

%% file: content.tex
\section{Introduction}

Boolean networks are adopted as modelling tools to organise knowledge and explore possible
behaviours emerging in biological processes~\cite{Schwab2021,zanudo21,montagud22}.
From the logic describing the influence between species, dynamics are defined to express the
evolution of variables at play. Update schemes that implement asynchrony are of particular interest,
as they express one form of inherent stochasticity~\cite{stoll2017}.
Attractors are fundamental structures that ought to capture the fate or stable behaviours of modelled systems.
Recently, the topic of identification of attractors of asynchronous dynamics has seen renewed interest and several developments.
To look at the last two years alone,
\cite{rozum2021parity} suggested an algorithm that combines different techniques such as motif detection and time reversal;
\cite{benevs2021computing} developed a symbolic approach that can handle large transition systems,
adopting binary decision diagrams to represent Boolean networks, and supporting partially defined update functions \cite{benevs2022aeon};
\cite{van2021improved,trinh2022computing} described approaches based on feedback vertex set identification and model checking reachability analysis.
The new techniques have enabled the handling of models of increasing complexity, as summarized in \cite{trinh2022computing}.

With this work we want to investigate the usefulness of network reduction in the investigation of attractors of asynchronous state transition graphs.
We refer specifically to the popular reduction method that consists in the iterative elimination of non-autoregulated components, as described in \cite{naldi2009reduction,naldi2011dynamically,veliz2011reduction}.
When a variable is selected for elimination, the regulatory functions of its targets are changed to remove the selected variable and replace it with its regulatory function.
The asynchronous dynamics and regulatory structures of the original network and the network obtained with this reduction process are related in a useful way; for instance, the networks have the same steady states, and the number of attractors of the reduced network cannot be smaller than the original one.
Properties of network reduction have already been exploited for attractor identification, in particular in \cite{rozum2021parity}, where reduction is used in combination with other methods under some specific cases. Here we show that a systematic use of variable elimination, adopted as the first step of the network analysis, can be quite beneficial.

The standing of variable elimination as a useful tool for identification of asynchronous attractors comes from the following property. For each attractor $\A$ of a Boolean network $f$, if $\tf$ is obtained from $f$ by iteratively eliminating some non-autoregulated components, there exists at least one state in an attractor of $\tf$ from which a state in $\A$ can be reconstructed, by tracing the reduction process backwards.
This property allows us to identify, from the attractors of $\tf$, a limited set of ``candidate''
states that cover all the attractors of $f$ (\cref{sec:reduction}).
Then, a screening of these candidate states is required for filtering out those outside the attractors of
$f$.
By calculating the minimal trap spaces \cite{trinh2022minimal,moon2022computational} we can first check if a candidate state is part of an attractor which is contained in a minimal trap space, and is the sole attractor contained in that minimal trap space.
If there is more than one candidate attractor in a minimal trap space, or if there are candidate
states outside of minimal trap spaces, other checks are required to determine whether they are actually part of an attractor. More computationally demanding techniques, for example from model checking, are useful for this step \cite{Pint,van2021improved} to check if a previously identified attractor, or a trap space not containing the candidate state, can be reached from the candidate state.

We tested these ideas on both biological and random networks by implementing variable elimination using \texttt{colomoto}'s \texttt{minibn} \cite{naldi2018colomoto}. The reduced network were then studied using AEON \cite{benevs2022aeon} and mtsNFVS \cite{trinh2022computing}, to identify candidate attractor states.
We used \texttt{trappist} to find the minimal trap spaces \cite{trinh2022minimal} and enable the first screening of the candidate states.
For the check of the remaining candidate states we used mtsNFVS's reachability analysis software \cite{trinh2022computing}.
We found, in general, high potential for reduction of computational times (\cref{sec:results}).
In addition, we observed that the reachability check, which is the most computationally expensive task in the pipeline, needs to be invoked only occasionally.
In fact, although variable elimination can lead to an increased number of attractors, this appears to happen in a quite limited number of cases, meaning that the number of candidate states is often very close, if not equal, to the number of attractors. As a result, the amount of work necessary to handle such situations is also contained.
In the last section (\cref{sec:discussion}) we comment on some avenues that could be explored for possible further improvements of the techniques discussed.

\section{Background}

\subsection{Boolean networks}\label{sec:background_bn}

We call $V$ the set of variables or species of interest, and set $n = |V|$.
A Boolean network is defined by a map $f = (f_1, \dots, f_n) \colon 2^n \to 2^n$.
The set $2^n$ is called the set of states or configurations of the Boolean network.
For $i=1,\dots,n$ and $x \in 2^n$, we denote by $\bar{x}^i$ the state that coincides with $x$ on all components except $i$.
A \emph{subspace} is a subset of the state space consisting of all the states that share the same values for a set of components, called the \emph{fixed} variables of the subspace.

The Boolean function $f_i$, $i=1,\dots,n$, is sometimes called the \emph{update function} of variable $i$.
The \emph{influence graph} $G(f)$ of $f$ is the signed multi-digraph with set of nodes $\{1,\dots,n\}$ and edges capturing the dependence of update functions on each variable: there exist an edge $(i, j)$ in $G(f)$ with sign $s$ if and only if, for some $x \in 2^n$, $f_j(\bar{x}^i) \neq f_j(x)$ and $s = (f_j(\bar{x}^i) - f_j(x))(\bar{x}^i_i - x_i)$.

A \emph{state transition graph} or \emph{dynamics} associated to $f$ is a graph with set of vertices $2^n$ and set of edges (also called \emph{transitions}) that depends on the chosen semantics.
This work deals with the \emph{asynchronous dynamics}, a form of non-deterministic dynamics which
includes only transitions between states differing by one component. Specifically, an edge exists from a state $x$ to a state $\bar{x}^i$ if and only if $f_i(x) \neq x_i$. We will write $\AD(f)$ to denote the asynchronous state transition graph of $f$.

A \emph{trap set} for the dynamics is a subset of the state space that does not admit any outgoing transition. An \emph{attractor} is a trap set that is minimal under inclusion.
Attractors can consist of singleton states, which are called \emph{steady states} or \emph{fixed points}, or can involve more than one state, in which case they are referred to as \emph{cyclic} or \emph{complex attractors}.

Determining whether a given state belongs to an attractor is a PSPACE-complete problem with
asynchronous or synchronous dynamics~\cite{PS2022}. In practice, computations usually rely on
computing (partly and symbolically) the state transition graph, which can be significantly more
complex in the asynchronous than in the synchronous case.

\emph{Trap spaces} are subspaces that do not admit any outgoing transition. They can be viewed as generalizations of steady states, since they are fixed points for the restriction of the dynamics to some components.
Compared to attractor identification, computation of trap spaces and related properties is a much more
tractable problem, due to a largely reduced theoretical complexity~\cite{moon2022computational}.
Minimal trap spaces are particularly interesting, since each of them must contain at least one attractor. Often minimal trap spaces of Boolean networks that serve as biological models contain only one attractor \cite{klarner2015approximating}, meaning, for instance, that they can replace the attractor they contain as reachability targets in reachability analysis.
Although some properties of the regulatory structure that can guarantee a ``good'' attractor landscape (attractors only in minimal trap spaces, uniqueness of attractors in minimal trap spaces) have been identified \cite{richard2023attractor,naldi2022linear}, these have limited application. In general, establishing whether a Boolean network admits multiple attractors in a minimal trap space, or attractors outside of minimal trap spaces (the ``motif-avoidant'' attractors of \cite{rozum2021parity}), remains a difficult task.

Based on these considerations, we can lay out the following classification of attractors of Boolean asynchronous dynamics into four categories, which will be useful in our discussion later:
\begin{itemize}
  \item[(A)] \emph{steady states}: these are ``easy'' to find. Since for our algorithm we calculate all minimal trap spaces, we find the steady states as the minimal trap spaces where all variables are fixed variables.
  \item[(B)] \emph{minimal univocal}: we use the term that was introduced in \cite{klarner2015approximating} to refer to minimal trap spaces that contain only one attractor. For convenience, we apply the term also to the attractors contained in these minimal trap spaces.
  \item[(C)] \emph{minimal nonunivocal}: these are attractors that are contained in a minimal trap space but are not the only attractor contained in that trap space.
To investigate their existence we have the convenience of being allowed to restrict the search space to the minimal trap space.
  \item[(D)] \emph{nonminimal} or \emph{motif-avoidant} \cite{rozum2021parity}: these are the most difficult to detect or exclude the existence of.
\end{itemize}

\begin{example}\label{ex:fgh}
The asynchronous state transition graphs of the maps
\begin{equation*}
  \begin{aligned}
    f(x_1,x_2) & = (x_1x_2 \vee \bar{x}_1\bar{x}_2, 0), \\
    g(x_1,x_2,x_3) & = (x_2 \bar{x}_1 \vee x_1 \bar{x}_2,
                        x_1 (x_2 x_3 \vee \bar{x}_2 \bar{x}_3) \vee \bar{x}_1 ( x_2 \bar{x}_3\vee x_3 \bar{x}_2),
                        x_2 x_3 \vee \bar{x}_2 \bar{x}_3), \\
    h(x_1,x_2) & = (x_1\bar{x}_2 \vee \bar{x}_1 x_2, x_1\bar{x}_2 \vee \bar{x}_1 x_2)
  \end{aligned}
\end{equation*}
are represented in \cref{fig:fgh}.

\begin{figure}[t!]
\begin{center}
\tikzcdset{arrow style=tikz, diagrams={>=stealth'}}
\begin{tikzpicture}[mybox/.style={draw, inner sep=5pt}]

\node[mybox,scale=1] at (-4,0) {
\begin{tikzcd}
   01 \arrow[d] & 11 \arrow[d] \\
   00 \arrow[r,yshift=2pt] & 10 \arrow[l,yshift=-2pt]
\end{tikzcd}
};

\node[mybox,scale=0.6] {
\begin{tikzcd}[row sep=small,column sep=small]
011 \arrow[rrr,yshift=1.5pt] \arrow[rd,yshift=-1.5pt,xshift=-1.5pt] & & & 111 \arrow[lll,yshift=-1.5pt] \\
& 001 \arrow[d,xshift=-1.5pt] \arrow[lu,yshift=+1.5pt,xshift=-1.5pt] & 101 \arrow[d,xshift=-1.5pt] & \\
& 000 \arrow[u,xshift=1.5pt] & 100 \arrow[rd,yshift=-1.5pt,xshift=1.5pt] \arrow[u,xshift=1.5pt] & \\
010 \arrow[rrr,yshift=1.5pt] & & & 110 \arrow[lll,yshift=-1.5pt] \arrow[lu,xshift=1.5pt,yshift=1.5pt]
\end{tikzcd}
};

\node[mybox,scale=1] at (4,0) {
\begin{tikzcd}
   01 \arrow[r,yshift=2pt] & 11 \arrow[l,yshift=-2pt] \arrow[d,xshift=-2pt] \\
   00 & 10 \arrow[u,xshift=2pt]
\end{tikzcd}
};

\end{tikzpicture}
\end{center}
\caption{Asynchronous state transition graphs of the Boolean networks $f$, $g$ and $h$ of \cref{ex:fgh}.
$\AD(f)$ has one minimal univocal attractor,
$\AD(g)$ has two minimal nonunivocal attractors,
and $\AD(h)$ has a minimal univocal and a nonminimal attractor.}\label{fig:fgh}
\end{figure}
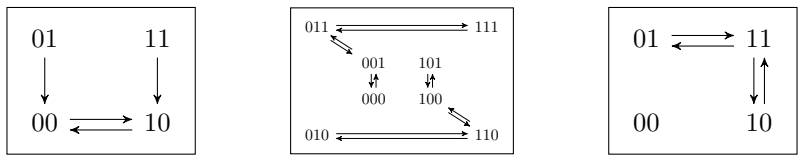

$\AD(f)$ has a cyclic attractor and minimal trap space $\{00,10\}$.

$\AD(g)$ has two attractors that are nonunivocal, since they are found in the same minimal trap space (the full space).

$\AD(h)$ has two attractors, one steady state 00 and one nonminimal attractor $\{01,10,11\}$. The steady state is the unique minimal trap space.
\end{example}

In the next section we review how the removal of components works in Boolean networks and discuss how it can help with identification of attractors.

\subsection{Reduction}\label{sec:reduction}

A popular reduction method for asynchronous dynamics of Boolean networks iteratively eliminates variables that are not autoregulated \cite{naldi2009reduction,naldi2011dynamically,veliz2011reduction}.
The approach has been recently extended to negatively autoregulated components \cite{schwieger2023reduction}. Although all our observations here can be extended, \textit{mutatis mutandis}, to negatively autoregulated components, we only discuss the standard case for sake of simplicity.

Suppose that $G(f)$ has a node that does not have a loop. Without loss of generality, we can assume that the node is $n$. We write $\pi \colon 2^n \to 2^{n-1}$ for the projection on the first $n-1$ variables.

By definition of $G(f)$, we have that $f_n(x,0)=f_n(x,1)$ for all $x \in 2^{n-1}$.
In particular, the state transition graph of $f$ admits exactly one of the transitions from $(x, 0)$ to $(x, 1)$ or from $(x, 1)$ to $(x, 0)$. We can therefore define a map $\S^n$ that associates to a ``reduced'' state a state in the larger space, as follows:
\begin{equation*}
  \begin{aligned}
    \S^n \colon 2^{n-1} & \to 2^n \\
    x & \mapsto (x, f_n(x, 0)) = (x, f_n(x, 1)).
  \end{aligned}
\end{equation*}

The reduction $\tf$ of $f$ obtained by elimination of component $n$ is then defined as
\begin{equation*}
  \tf = (f_1 \circ \S^n, \dots, f_{n-1} \circ \S^n) \colon 2^{n-1} \to 2^{n-1}.
\end{equation*}
Given $y \in 2^n$, the state $\S^n(\pi(y))$ can be thought of as the ``representative state'' of the pair $(\pi(y), 0)$, $(\pi(y), 1)$, or the state that ``survives the reduction'', since all transitions leaving the state $\S^n(\pi(y))$ have a corresponding transition in $\AD(\tf)$, whereas the transitions leaving the state $\overline{\S^n(\pi(y))}^n$ are not guaranteed to be preserved \cite{naldi2011dynamically}.

It was shown in \cite{naldi2009reduction,naldi2011dynamically} that $f$ and $\tf$ admit the same number of steady states, and the number of attractors of $\tf$ is greater or equal to the number of attractors of $f$.
The following result, which forms the basis for our method, is a simple consequence of properties proved in \cite{naldi2011dynamically}.

\begin{theorem}
  If $\A$ is an attractor of $f$, then there exists at least one attractor for $\tf$ in $\pi(\A)$, and for each $x \in \pi(\A)$ contained in an attractor of $\tf$, $\S^n(x) \in \A$.
\end{theorem}
\begin{proof}
  The first part is a consequence of the observation that, if $B$ is a trap set for $\AD(f)$, then $\pi(B)$ is a trap set for $\AD(\tf)$.

  Given a state $x$ in an attractor for $\tf$ contained in $\pi(\A)$, by definition of $\pi$ either $\S^n(x)$ or $\overline{\S^n(x)}^n$ is in $\A$, and since $\S^n(x)$ is reachable from $\overline{\S^n(x)}^n$, $\S^n(x)$ must be in $\A$.
\end{proof}

\begin{figure}[t!]
\centering
\input{attractor-illustration}
\caption{Idea behind the approach: states in attractors of reduced networks $\tf$ can be used to find candidate states in attractors of $f$.}\label{fig:idea}
\end{figure}
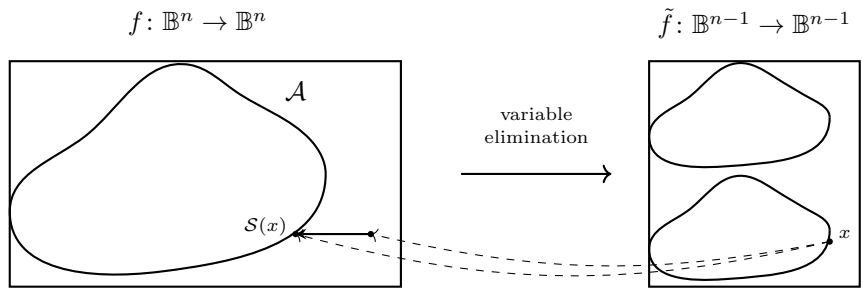

The theorem gives us the following property: if there exists an attractor for $f$, then we are able to identify one state of this attractor by finding the attractors of the reduced version $\tf$, sampling a state from each of these attractors, and ``lifting'' the states to the original space using $\S^n$ (\cref{fig:idea}).

Before we discuss more in detail how we can use the reduced network to identify attractors of the original network, let us look at some examples,
which can give an idea of what can happen to attractors with variable elimination.

\begin{example}\label{ex:Sx-needed}
Consider the map $f$ from \cref{ex:fgh}.
Since the second variable is not autoregulated, it can be removed using the elimination method described in this section. The constant value $0$ replaces $x_2$ in the update function of $x_1$, giving the reduced network $\tf(x_1) = \bar{x}_1$.
We have $\S^2(0)=\S^2(00)$ and $\S^2(1)=\S^2(10)$.
This simple example is sufficient to illustrate how, given a state $x$ in the attractor of a reduced network (say $x=0$ here), to retrieve a state in an attractor of the original network we cannot pick any state in the preimage $\pi^{-1}(x)$ (the state $01$ does not work), but we need to apply the function $\S^2$.
\end{example}

\begin{figure}[t!]
\resizebox{\textwidth}{!}{
\tikzcdset{arrow style=tikz, diagrams={>=stealth'}}
\begin{tikzpicture}[mybox/.style={draw, inner sep=5pt}]

\node[mybox,scale=1.1] at (-8,0) {
\begin{tikzcd}[row sep=small,column sep=small,ampersand replacement=\&]
011 \arrow[dr,yshift=1] \& \& \& 111 \arrow[lll] \arrow[ld] \arrow[ddd] \\
\& 001 \arrow[lu,yshift=-1] \arrow[d] \arrow[r,yshift=-1] \& 101 \arrow[l,yshift=1] \& \\
\& 000 \arrow[r,yshift=1] \& 100 \arrow[u] \arrow[l,yshift=-1] \arrow[rd,yshift=-1]  \& \\
010 \arrow[ru] \arrow[uuu] \arrow[rrr] \& \& \& 110 \arrow[ul,yshift=1]
\end{tikzcd}
};

\node[mybox,scale=0.893] {
\begin{tikzcd}[row sep=small,column sep=small,ampersand replacement=\&]
0110 \arrow[rrr] \arrow[dr,yshift=1] \& \& \& 1110 \arrow[dl,yshift=1] \arrow[rrrr,bend left=15] \& 0111 \arrow[llll,bend right=15] \& \& \& 1111 \arrow[lll] \\
\& 0010 \arrow[d,xshift=1] \arrow[ul,yshift=-1] \& 1010 \arrow[rrrr,bend left=15] \arrow[ur,yshift=-1] \arrow[d,xshift=1] \arrow[l] \& \& \& 0011 \arrow[d,xshift=-1] \arrow[r] \arrow[llll,bend right=15] \& 1011 \arrow[d,xshift=-1] \&\\
\& 0000 \arrow[u,xshift=-1] \& 1000 \arrow[l] \arrow[u,xshift=-1] \arrow[rrrr,bend right=15] \& \& \& 0001 \arrow[r] \arrow[dl,yshift=1] \arrow[llll,bend left=15] \arrow[u,xshift=1] \& 1001 \arrow[u,xshift=1] \arrow[dr,yshift=1] \&\\
0100 \arrow[rrr] \& \& \& 1100 \arrow[rrrr,bend right=15] \& 0101 \arrow[ur,yshift=-1] \arrow[llll,bend left=15] \& \& \& 1101 \arrow[lll] \arrow[ul,yshift=-1]
\end{tikzcd}
};

\node[mybox,scale=1.1] at (8,0) {
\begin{tikzcd}[row sep=small,column sep=small,ampersand replacement=\&]
011 \arrow[dr] \arrow[ddd] \& \& \& 111 \arrow[dl] \arrow[lll] \\
\& 001 \arrow[d] \& 101 \arrow[d] \arrow[l] \& \\
\& 000 \& 100 \arrow[dr] \& \\
010 \arrow[rrr] \& \& \& 110 \arrow[uuu]
\end{tikzcd}
};
\end{tikzpicture}
}
\caption{Asynchronous state transition graphs of the Boolean networks $\hat{f}$, $\hat{g}$, $\hat{h}$ of \cref{ex:attractors-can-split,ex:attractors-can-increase,ex:not-all-states}.}\label{fig:exs}
\end{figure}
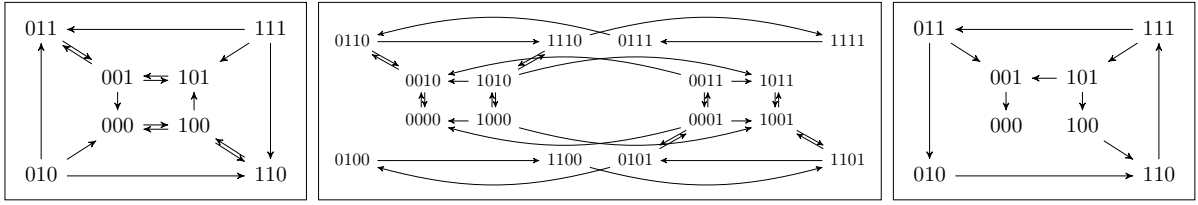

\begin{example}\label{ex:not-all-states}
If $\A$ is an attractor for a Boolean network and $x$ is a state in $\A$, does $\pi(x)$ necessarily belong to an attractor of the reduced network?
The answer is negative.
Take for instance
$$\hat{f}(x_1,x_2,x_3) =
(\bar{x}_1 \bar{x}_2 \vee \bar{x}_1 \bar{x}_3 \vee x_2 \bar{x}_3,
x_1 \bar{x}_2 \bar{x}_3 \vee \bar{x}_1 \bar{x}_2 x_3,
x_1 \bar{x}_2 \vee \bar{x}_1 x_2).$$
By removing the third component, we obtain the function $f$ of \cref{ex:fgh}.
The state $011$ is in the unique attractor of $\hat{f}$, while $\pi(011)=01$ does not belong to any attractor of $f$.

In general, therefore, we are not able to retrieve all states of an attractor of a Boolean network by lifting states in attractors of its reduction.
We will only find some states, from which we can visit the attractor if required.
\end{example}

\begin{example}\label{ex:attractors-can-split}
The Boolean network
$$\hat{g}(x_1,x_2,x_3,x_4) = (x_2 \bar{x}_4 \vee \bar{x}_2 x_4,
                              x_4 (x_2 x_3 \vee \bar{x}_2 \bar{x}_3) \vee \bar{x}_4 (x_2 \bar{x}_3 \vee x_3 \bar{x}_2),
                              x_2 x_3 \vee \bar{x}_2 \bar{x}_3,
                              x_1)$$
has one cyclic attractor, that fills the whole state space (see \cref{fig:exs}).
By removing variable $x_4$ we obtain the network $g$ in \cref{ex:fgh},
which has two attractors.
\end{example}

\begin{example}\label{ex:attractors-can-increase}
The Boolean network
$$\hat{h}(x_1, x_2, x_3) = (x_1\bar{x}_3 \vee x_2\bar{x}_3, x_1\bar{x}_3 \vee x_2\bar{x}_3, x_1x_2)$$
has a unique attractor, the fixed point 000.
By eliminating variable $x_3$ we obtain the Boolean network $h$ in \cref{ex:fgh}, which has two attractors.
\end{example}

Suppose that $\tf \colon 2^m \to 2^m$ is obtained from $f$ by iteratively eliminating variables $n, n-1, \dots, m+1$, and that $\A_1,\dots,\A_M$ are attractors of $\tf$.
Take one state $x^1,\dots,x^M \in 2^m$ in each attractor. We can reconstruct the corresponding states in $2^n$ by applying the map $\S = \S^n \circ \S^{n-1} \circ \cdots \circ \S^{m+1}$. How can we establish whether each of these states is in an attractor of $f$, and how many attractors $f$ has?

We can calculate the minimal trap spaces of $f$, and make the following observations.
Given a set of candidate states $C = \{\S(x^1), \dots, \S(x^M)\}$:

\begin{itemize}
  \item[(a)] If $x$ is a steady state of $\tf$, then $\S(x)$ is a steady state of $f$ (and all steady states of $f$ can be calculated in this fashion).
  \item[(b)] if $\S(x)$ belongs to a minimal trap space $t$ of $f$, and is the only state in $C$ that belongs to $t$, then $\S(x)$ belongs to an attractor of $f$ that is minimal univocal.
We call these \emph{univocal states}.
  \item[(c)] if $\S(x)$ is contained in a minimal trap space $t$ of $f$, and is not the only state in $C$ that is contained in this minimal trap space, then we need to study the dynamics in $t$ to clarify whether each candidate state contained in $t$ belongs to an attractor, and whether the candidate states in $t$ belong to different attractors. We can call these states \emph{candidate nonunivocal}. The number of states in $C$ contained in $t$ gives an upper bound on the number of attractors contained in $t$.
  \item[(d)] if $\S(x)$ is not contained in any minimal trap space, then $\S(x)$ \emph{might} belong to a nonminimal attractor of $f$. To establish whether this is the case, and to find the number of nonminimal attractors, we need to do additional work. We refer to these states as \emph{candidate nonminimal}.
\end{itemize}

Note that these observations would have to be slightly changed were one to consider the elimination of negatively autoregulated components \cite{schwieger2023reduction}.
We can now give the description of the approach to the identification of attractors based on elimination of components.

\section{Method}

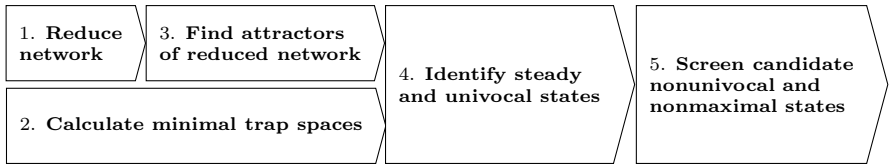
\begin{figure}[t!]
  \centering
  \tikzset{resarrow/.style={draw,
  minimum height=1.0cm,
  inner sep=0.5em,
  shape=signal,
  signal to=east,
  signal pointer angle=150,
  text centered,
  draw=black, fill=white,
  font=\scriptsize
  }
  }
  \tikzset{tallarrow/.style={draw,
  minimum height=2.1cm,
  inner sep=0.5em,
  shape=signal,
  signal to=east,
  signal pointer angle=150,
  text centered,
  draw=black, fill=white,
  font=\scriptsize
  }
  }
  \begin{tikzpicture}
    \node[resarrow, align=left] (c) at (2,0) {1. \textbf{Reduce}\\\textbf{network}};
    \node[resarrow, align=left] (b) at (4.52,0) {3. \textbf{Find attractors}\\
                                                \textbf{of reduced network}};
    \node[resarrow, align=left] (a) at (3.6,-1.1) {2. \textbf{Calculate minimal trap spaces\hspace{0.0cm}}};
    \node[tallarrow, align=left] (a) at (7.7,-0.55) {4. \textbf{Identify steady}\\
                                                    \textbf{and univocal states}};
    \node[tallarrow, align=left] (a) at (11.05,-0.55) {5. \textbf{Screen candidate}\\
                                                    \textbf{nonunivocal and}\\
                                                    \textbf{nonmaximal states}};
  \end{tikzpicture}\caption{Main steps of the algorithm.}\label{fig:steps}
\end{figure}

Based on the analysis of the relationship between attractors of reduced versions of a Boolean network and attractors of the original network, we propose the following pipeline (see \cref{fig:steps}).

First reduce the network $f$ by eliminating variables (step 1), and, possibly in parallel, find the minimal trap spaces of the Boolean network $f$ (step 2).
Then, identify one state for each attractor of the reduced the network $\tf$ that is not a steady state (step 3), obtaining a set of candidate attractor states for $f$.

Step 4 and 5 deal with the screening of these candidate attractor states.
Step 4 is the easy part: we check whether each candidate state is contained in a minimal trap space. If a minimal trap space contains only one candidate state, then this is a univocal state and we identified a univocal attractor for $f$ (point (b) of the last section).

Step 5 takes care of the remaining candidate states.
These can be either multiple states contained in the same minimal trap space (point (c) of the last section), or states that do not belong to any minimal trap space (point (d) of the last section).
In both scenarios, we have to study in some way the state transition graph to understand whether the candidate states belong to an attractor, and how many attractors exist (some candidate states might be part of the same attractor, which has been ``split'' during the reduction process, like in \cref{ex:attractors-can-split}, or part of an attractor that was created by the reduction process \cref{ex:attractors-can-increase}). For this step one can use model checking approaches (see the next section), possibly combined with other techniques (see the discussion in the last section).

\subsection{Implementation}\label{sec:implementation}
\subsubsection{Software}
For our implementation, we use the Python library \texttt{colomoto}'s \texttt{minibn}
\cite{naldi2018colomoto} to compute the network reduction (step 1 in \cref{fig:steps}).
In \texttt{minibn}, the local functions of the Boolean network are represented in propositional
logic formula, with usual Boolean algebra.
Given a component $i$ to reduce, our implementation simply substitutes $x_i$ with the expression of $f_i$ in
all its targets. Then, basic Boolean expression simplifications are performed, which may result in
variable elimination.

For step 2, we use \texttt{trappist} to calculate the minimal trap spaces \cite{trinh2022minimal}.
Because \texttt{trappist} relies on a Petri net representation of the Boolean network, this step
involves a transformation of each local Boolean function $f_i$ in two expressions in DNF form
(one for $f_i(x)=1$ and one for $f_i(x)=0$), from which are derived the Petri net transitions.

For step 3, the identification of attractors of the reduced network, we consider two recently developed methods, AEON \cite{benevs2022aeon} and mtsNFVS \cite{trinh2022computing}. In the analysis of \cite{trinh2022computing}, they have been shown to be the two fastest methods available, while implementing two very different approaches.
The outputs of the two methods are also quite different: AEON, for instance, can give information on the attractors (cardinality, list of states) that is not directly available with mtsNFVS, and can deal with families of networks at the same time, do postprocessing control tasks, etc. Here we are only looking at the performance of the methods in regard to the identification of the number of attractors and their nature (steady or cyclic), and we are mostly concerned with understanding whether reduction might be useful in this respect.

The pieces are put together and step 4 performed in \texttt{python}.
For step 5, we run the \texttt{java} tool \texttt{mtsNFVS.jar} available as part of mtsNFVS.
The tool uses \texttt{Pint} \cite{Pint} for a first check via static analysis followed by bounded and exact model checking \cite{van2021improved}.
We noticed however that mtsNFVS (without reduction) produces wrong results in several instances (some examples are provided in our project repository).
Hence, the results obtained with mtsNFVS should be taken with caution.
For the biological networks we considered, we checked that all methods found the same number of steady states and cyclic attractors.
As we will discuss in the Results section, the reachability analysis step has limited impact on the overall process, since nonminimal candidates appear only rarely, and no nonunivocal candidates are found in any of the biological or random tests.

\subsubsection{Elimination order and other considerations}

The choice of order for the variables to be eliminated has an influence on the reduced network as well as on the running times.
When a variable with $r$ regulators and $t$ targets is eliminated, the $r + t$ regulations could be replaced, in the worst case scenario, by $r \cdot t$ regulations. We therefore pick for elimination, at each step, one of the variables for which this product is minimum. A more systematic study of how the elimination order influences the running time could lead to identification of better elimination orders. For instance, one might try to favour elimination orders that have the least impact on the number of trap spaces.

Another crucial point when using variable elimination is deciding where to stop the iterative process. Unfortunately, there does not seem to be a simple answer to this question. In some cases, at some point in the elimination process reduction can introduce complicated update functions that can slow the processing steps that follow.
In addition, although reduction seems to be beneficial as a preliminary step for the methods we tested, we will see that different levels of reduction might favour different methods variably.
In our implementation we considered two possible parameters to stop the reduction process. One specifies the minimum size for the reduced network in terms of number of nodes (\texttt{stop\_at}). The other looks at the minimum product of number of regulators and number of targets for each node (that is, the parameter used to choose the variable to eliminate), and sets a maximum for this value (\texttt{max\_product}).

We also investigated the impact of the simplification of Boolean expressions on running times. We found that, although simplification is expensive, it leads to faster elimination.
In the next section we will discuss the impact of the elimination step on running times for both biological and random networks.

\subsection{Results}\label{sec:results}

The experiments on the biological network models were run on a desktop computer with an Intel(R) Core(TM) i7-8700 processor, 32GB of RAM, operating system Debian GNU/Linux 11. 
The experiments on the random networks were run on an HPC cluster nodes with an AMD(R) Zen2 EPYC 7702
@ 2 GHz, 1TB of RAM, operating system, operating system CentOS Linux.
The code is available at \url{https://github.com/etonello/attractors-with-reduction}.

\subsubsection{Biological models}

\begin{table}[t!]
\centering
\begin{tabular}{ l l l c c c c}
\hline
\thead{Model} & \thead{file name} & \thead{file source} & \thead{n.\\nodes} & \thead{n.\\edges} & \thead{n. steady\\states} & \thead{n. cyclic\\attractors}\\
\hline
MAPK  & grieco\_mapk                 & PyBoolNet &  53 & 108 &    12 &    6 \\
IL-6  & IL\_6\_Signalling            & mtsNFVS   &  55 &  95 & 28672 & 4096 \\
EMT   & selvaggio\_emt               & PyBoolNet &  56 & 159 &  1452 &    0 \\
T-LGL & TLGLSurvival                 & mtsNFVS   &  58 & 193 &   172 &  146 \\
CACC  & Colitis\_associated\_\dots   & mtsNFVS   &  66 & 144 &     2 &    8 \\
AD    & A\_model                     & mtsNFVS   &  74 & 198 &     0 &    2 \\
AGS   & id-148-AGS-CELL-FATE-\dots   & biodivine &  83 & 193 &     1 &    0 \\
CC    & cell\_cycle\_2019            & mtsNFVS   &  87 & 467 &     8 &    0 \\
SP    & id-192-SEGMENT-POLA\dots     & biodivine & 102 & 432 &    65 &    0 \\
SIPC  & SIGNALING-IN-PROST\dots      & mtsNFVS   & 116 & 428 &  2460 &  300 \\
DSP   & id-210-DRUG-SYNERGY-\dots    & biodivine & 144 & 367 &     0 &    1 \\
C3.0  & CASCADE3                     & mtsNFVS   & 176 & 449 &     0 &    1 \\
EP    & id-211-EPITHELIAL-DER\dots   & biodivine & 183 & 602 &     0 &    1 \\
\hline
\end{tabular}
\caption{Information on the sourcing of the bnet models considered, their size and the number of steady states and cyclic attractors.}\label{tab:bio-info}
\end{table}

\begin{table}[t!]
\centering
\begin{tabular}{ l c c c c c c c c c c c}
\hline
              & \multicolumn{3}{c}{(1) \texttt{max\_product=20}} & & \multicolumn{3}{c}{(2) \texttt{max\_product=50}} & & \multicolumn{3}{c}{(3) \texttt{max\_product=100}} \\\cline{2-4}\cline{6-8}\cline{10-12}
\thead{Model} & n. nodes & n. edges & time         & & n. nodes & n. edges & time         & & n. nodes & n. edges & time \\
\hline
MAPK  &   10 &   41 &    0.1 &   &   10 &   41 &    0.0 &   &   10 &   41 &    0.1 \\
IL-6  &   17 &   28 &    0.0 &   &   17 &   28 &    0.0 &   &   17 &   28 &    0.0 \\
EMT   &   19 &   94 &    0.1 &   &   17 &  130 &    0.1 &   &   17 &  130 &    0.1 \\
T-LGL &   21 &   91 &    0.0 &   &   18 &  111 &    0.0 &   &   18 &  111 &    0.0 \\
CACC  &   14 &   81 &    0.0 &   &   11 &   56 &    0.0 &   &   11 &   56 &    0.0 \\
AD    &   11 &   93 &    0.0 &   &   10 &   97 &    0.0 &   &   10 &   97 &    0.0 \\
AGS   &    2 &    4 &    0.0 &   &    2 &    4 &    0.0 &   &    2 &    4 &    0.0 \\
CC    &   38 &  415 &    0.2 &   &   35 &  490 &    0.3 &   &   29 &  669 &   23.9 \\
SP    &   35 &  234 &    0.0 &   &   33 &  273 &    0.1 &   &   32 &  357 &    0.1 \\
SIPC  &   41 &  390 &    0.3 &   &   32 &  465 &    0.6 &   &   28 &  522 &    8.7 \\
DSP   &   14 &  108 &    0.0 &   &   10 &   71 &    0.0 &   &   10 &   71 &    0.0 \\
C3.0  &   23 &  191 &    0.0 &   &   14 &  175 &    0.1 &   &   13 &  193 &    0.1 \\
EP    &   33 &  322 &    0.1 &   &   25 &  365 &    0.1 &   &   21 &  316 &    0.1 \\
\hline
\end{tabular}
\caption{Reduction scenarios considered for the biological models listed in \cref{tab:bio-info}.
\texttt{max\_product}$=k$ indicates that the elimination ends when the product of number of regulators and number of targets is above $k$ for all nodes. The elimination is also set to stop when the number of nodes reaches 10.}\label{tab:scenarios}
\end{table}

\begin{table}[t!]
\centering
\begin{tabular}{  l c c c c c c c c c }
    \hline
              & \multicolumn{4}{c}{AEON running times} & ~ & \multicolumn{4}{c}{mtsNFVS running times}
              \\\cline{2-5}\cline{7-10}
    \thead{Model} & No red. & (1) & (2) & (3) && No red. & (1) & (2) & (3) \\
\hline
MAPK  &    5.7 &    0.3 &    0.3 &    0.3 & &    28.9 $\scriptstyle\pm   5.7$ (3 DNF) &    0.7 &    0.7 &    0.7 \\
IL-6  &  774.6 &    6.0 &    6.0 &    6.0 & &    14.8 $\scriptstyle\pm   1.8$ &    7.4 &    7.4 &    7.4 \\
EMT   &   25.6 &    0.8 &    0.7 &    0.7 & &  DNF &    1.3 &    1.4 &    1.4 \\
T-LGL &   17.5 &    0.9 &    0.9 &    1.1 & &     2.2 &    1.8 &    1.9 &    2.7 \\
CACC  &    9.3 &    0.3 &    0.3 &    0.3 & &     0.5 &    0.8 &    0.7 &    0.7 \\
AD    &  361.9 &    0.3 &    0.4 &    0.4 & &     0.7 &    1.0 &    0.8 &    0.8 \\
AGS   &    1.7 &    0.3 &    0.3 &    0.3 & &     0.6 &    0.7 &    0.7 &    0.7 \\
CC    &    DNF &   27.0 &   11.1 &   26.9 & &     8.2 $\scriptstyle\pm   3.5$ &    2.2 &    6.0 &   39.4 \\
SP    &    DNF &    0.9 &    0.8 &    0.8 & &  DNF &    1.0 &    1.1 &    1.4 \\
SIPC  &    DNF &   28.2 &    6.9 &   14.4 & &  1664.7 $\scriptstyle\pm 506.3$ &   47.8 $\scriptstyle\pm   5.6$ &   53.7 $\scriptstyle\pm   7.4$ &  138.8 \\
DSP   &    DNF &    0.4 &    0.4 &    0.4 & &     2.3 &    0.8 &    0.7 &    0.7 \\
C3.0  &    DNF &    0.5 &    0.4 &    0.4 & &     2.1 &    1.0 &    1.0 &    0.9 \\
EP    &    DNF &    1.5 &    0.6 &    0.5 & &    62.7 $\scriptstyle\pm  64.2$ &    2.5 &    2.4 &    1.9 \\
\hline
\end{tabular}
\caption{Running times for AEON \cite{benevs2022aeon} and mtsNFVS \cite{trinh2022computing} without reduction and for the three reduction scenarios
described in \cref{tab:scenarios}. The running times in the three reduction scenarios include the
time for network reduction. Running times for mtsNFVS are averaged over five iterations,
and show the standard deviation if this is higher than 1s.
``DNF'' indicates that the processing did not complete within one hour.}\label{tab:times}
\end{table}

We run our implementation using AEON \cite{benevs2022aeon} and mtsNFVS \cite{trinh2022computing} to find the attractors of the reduced networks, and compared the running times to those of AEON and mtsNFVS applied on the networks without reduction.
We consider the seven models examined in \cite{trinh2022computing} as available in the tool repository, as well as additional models extracted from the PyBoolNet repository \cite{klarner2017pyboolnet} and the biodivine-boolean-models repository (\url{https://github.com/sybila/biodivine-boolean-models}).
Information on the source of each model, as well as the size of the networks is given in \cref{tab:bio-info}.
We refer the reader to the respective sources for references detailing each model, which explain if and how the models have been modified from their original sources.

As we mentioned in the previous section, different levels of reduction can have different impacts on running times of both the elimination procedure and the remainder of the attractor identification process.
Here we report running times obtained in three scenarios, where we set the parameter \texttt{max\_product} described in the previous section (maximum value accepted for the minimum product of in- and out-regulations over all nodes) to three levels, 20 in scenario (1), 50 in scenario (2) and 100 in scenario (3).
In all three cases, we set the minimum size of the reduced network to 10.
In some cases these parameters lead to network with fewer than 10 nodes, because constants variables might be generated by the elimination process, and these are always eliminated.
In \cref{tab:scenarios} we show the number of nodes and signed regulations of the reduced network obtained in each of the three scenarios, as well as the time required by the reduction process. The time for reduction is below one second except for two networks, CC and SIPC, where it goes up to 24 and 9 seconds.

\cref{tab:times} shows the running times for the two methods and the different scenarios (no reduction, reduction scenario (1), (2) and (3)).
The algorithm of mtsNFVS relies on some random choices, which cause its running times to vary, sometimes significantly. We report the minimum and maximum running times obtained on five iterations. We set a timeout of one hour, and ``DNF'' indicates that the processing did not complete within this time.

It is clear from the results that approaching the attractor identification problem by first reducing the network can speed up the processing significantly.
We can also observe that a more aggressive reduction does not necessarily translate to faster attractor identification times.
This can be observed for both attractor identification methods for the networks CC and SIPC.
Looking at the numbers for these two networks, we can also see that the reduction levels giving the shortest processing times are not the same for AEON and mtsNFVS.
In particular, mtsNFVS seems to work better with more conservative reduction levels, and in one case (network CC, scenario (3)) the time required for the reduction process was higher than the running time for mtsNFVS without reduction.

Importantly, only for one of the networks (TLGL) nonminimal candidates were identified, and no network presented non-univocal candidates,
meaning that reduction had overall a limited impact on the general attractor configuration.
Looking at what happens on randomly-generated networks might help to clarify whether these behaviours could be specific of biological models or more general.

\subsubsection{Random benchmarks}

As in \cite{trinh2022computing}, we generate random networks by invoking \texttt{generateRandomNKNetwork} from the R package BoolNet (\cite{mussel2010boolnet}), fixing $K = 2$ for the number of regulators for each variable.
We consider networks with $n = 100 k$ nodes, with $k=1,\dots,5$, and create 10 networks for each size.

We tested several stopping conditions for the reduction process, and did not identify a general rule that would give optimal times in all cases.
To give an idea of the range of possible results, we report running times for three reduction scenarios, where we set the parameter \texttt{max\_product} to $\frac{n}{2}$, $n$ and $2n$,
while \texttt{stop\_at} is set to $\frac{n}{10}$.
The details of the three scenarios are shown in \cref{tab:scenarios-n-k}.
In \cref{tab:times-n-k}, we show the average running times for networks that were processed within the timeout of one hour, and the number of networks that were not processed in time.

We observe again that, by adopting the reduction approach, the number of networks that can be successfully analysed within the timeout given increases significantly.
AEON without reduction could process 4 out of the 10 networks of size 100, and no network of larger size.
With reduction and AEON, all networks of size 200 and 8 out of the 10 networks of size 300 could be processed,
as well as one network of size 400.
The running times are just a few seconds for networks of size 100, and vary significantly for larger networks.

The tool mtsNFVS without reduction could process all networks of size 100 and some networks of size 200 and 300.
With reduction, networks of up to size 500 could be handled.
However, as we pointed out in \cref{sec:implementation}, the results generated by mtsNFVS might require further validation, as several failures were identified in test networks.

We can nevertheless, by observing the results obtained with AEON, note that the number of candidate states that need processing via reachability analysis remains very limited.
Only five nonminimal candidate states were identified for the networks that were processed with reduction and AEON, and no nonunivocal candidate states.
This suggests that nonminimal and nonunivocal attractors might be rare phenomena, even for random networks.

Reduction allowed the size of networks to be drastically reduced, but the results shown in \cref{tab:scenarios-n-k} illustrate how different number of reduction steps might affect networks in different ways.
Some networks could only be processed in time in scenario (3), with the highest number of eliminated variables; for others, processing times were lower when a smaller number of variables was removed.

\begin{table}[t!]
\centering
\begin{tabular}{ c c c c c c c c c c c c}
\hline
              & \multicolumn{3}{c}{(1) \texttt{max\_product=n/2}} & & \multicolumn{3}{c}{(2) \texttt{max\_product=n}} & & \multicolumn{3}{c}{(3) \texttt{max\_product=2n}} \\\cline{2-4}\cline{6-8}\cline{10-12}
\thead{Model\\size} & \thead{average\\n. nodes} & \thead{average\\n. edges} & \thead{average\\time} && \thead{average\\n. nodes} & \thead{average\\n. edges} & \thead{average\\time} && \thead{average\\n. nodes} & \thead{average\\n. edges} & \thead{average\\time} \\
\hline
100 &    14.0 &   219.9 &     0.1 &&    12.3 &   213.0 &     0.5 &&    11.9 &   199.5 &     0.6 \\
200 &    24.0 &   621.9 &     0.5 &&    21.2 &   618.0 &     1.9 &&    20.7 &   596.3 &     5.8 \\
300 &    37.6 &  1418.5 &     2.1 &&    34.4 &  1527.5 &     7.8 &&    30.0 &  1240.0 &    19.7 \\
400 &    49.4 &  2266.0 &     2.6 &&    42.5 &  2471.0 &    23.1 &&    41.1 &  2498.3 &    55.3 \\
500 &    61.5 &  3489.8 &    18.4 &&    55.3 &  3985.7 &    18.8 &&    55.5 &  4683.5 &   286.5 \\
\hline
\end{tabular}
\caption{Statistics for three reduction scenarios on 10 random networks (generated with BoolNet \cite{mussel2010boolnet}, setting $K = 2$). \texttt{max\_product}$=k$ indicates that the elimination ends when the product of number of regulators and number of targets is above $k$ for all nodes. The elimination is also set to stop when one-tenth of the number of nodes of the original network is reached.}\label{tab:scenarios-n-k}
\end{table}

\begin{table}[t!]
\centering
\begin{tabular}{ c c c c c c c c c c}
\hline
& \multicolumn{4}{c}{AEON running times} \\\hline
\thead{Model\\size} & No red. & (1) & (2) & (3)\\
\hline
100 &   1104.6 $\scriptstyle\pm 676.0$        (6) &      2.3 $\scriptstyle\pm  1.3$ &      3.0 $\scriptstyle\pm  1.5$ &      3.2 $\scriptstyle\pm  1.5$ \\
200 &  (10) &     83.7 $\scriptstyle\pm 205.7$ &     48.9 $\scriptstyle\pm 113.7$ &    146.1 $\scriptstyle\pm 402.4$ \\
300 &  (10) &   1331.3 $\scriptstyle\pm 1394.7$        (5) &    696.0 $\scriptstyle\pm 684.1$        (3) &    888.9 $\scriptstyle\pm 1185.0$        (2) \\
400 &  (10) &         (10) &   2994.6        (9) &   3033.2        (9) \\
500 &  (10) &         (10) &         (10) &         (10) \\
\hline
\end{tabular}

\begin{tabular}{ c c c c c c c c c c}
\hline
& \multicolumn{4}{c}{mtsNFVS running times} \\\hline
\thead{Model\\size} & No red. & (1) & (2) & (3) \\
\hline
100 &      3.3 &      6.9 $\scriptstyle\pm  2.0$ &      7.6 $\scriptstyle\pm  2.1$ &      7.7 $\scriptstyle\pm  2.0$ \\
200 &    106.8 $\scriptstyle\pm 476.8$        (11) &    378.4 $\scriptstyle\pm 592.9$        (5) &    314.5 $\scriptstyle\pm 629.4$        (7) &    379.8 $\scriptstyle\pm 654.0$        (5) \\
300 &     71.7 $\scriptstyle\pm 90.5$        (25) &   1440.8        (49) &   2298.6 $\scriptstyle\pm 653.3$        (44) &   3305.8 $\scriptstyle\pm 432.1$        (46) \\
400 & (50) &   2059.3 $\scriptstyle\pm 644.0$        (26) &   2639.5 $\scriptstyle\pm 718.3$        (43) &   2239.0 $\scriptstyle\pm 482.2$        (41) \\
500 & (50) &   1672.3 $\scriptstyle\pm 562.8$        (44) &         (50) &   1276.6 $\scriptstyle\pm 418.3$        (46) \\
\hline
\end{tabular}
\caption{Average running times for AEON \cite{benevs2022aeon} and mtsNFVS \cite{trinh2022computing}, on random networks without reduction and in the three reduction scenarios of \cref{tab:scenarios-n-k}. In parentheses is the number of processes that did not terminate within one hour. Since running times for mtsNFVS have a high variability, we run each test five times. The mean and standard deviation shown are over all successful tests for the given size.}\label{tab:times-n-k}
\end{table}

\section{Discussion}\label{sec:discussion}

We investigated how reduction can make the process of attractor identification faster.
We observed that reduction generally makes the process easier, but the impact can depend on the adopted attractor detection approach.
Elimination of variables, while reducing the size of the network, increases the complexity of the influence graph and update functions.
This, on its turn, has a different impact on the two tools we tested, AEON \cite{benevs2022aeon} and mtsNFVS \cite{trinh2022computing},
which implement very different approaches to identification of attractors.
Deeper investigations targeted on the specific tool might give some insights on why certain levels of reduction work better than others.
At the same time, we did limited analysis on the impact of the order of elimination;
we cannot exclude that specific orders of elimination might be devised to better suit a specific method
(for instance, for mtsNFVS, one might want to investigate orders that do not increase the size of minimal feedback vertex sets).

Although candidate states that require reachability analysis are only rarely encountered,
this step might benefit from additional developments.
One improvement would come from the screening of nonunivocal candidate states,
as the technique of mtsNFVS currently can fail to detect the existence of multiple attractors contained in the same minimal trap space.
In addition, other techniques could be incorporated for the exclusion of existence of nonmaximal attractors.
For instance, at the moment only minimal trap spaces are used.
Larger trap spaces (the maximal trap spaces that do not contain the candidate state) could provide bigger targets for reachability analysis, while still being easily identifiable with a small modification to the approach of \texttt{trappist} \cite{trinh2022minimal}.

Finally, attractor detection tools are capable of other tasks, for instance, AEON can perform detection of bifurcations and source-target control.
There is the potential that reduction could be used sensibly to speed up these activities too.
Each task needs to be studied individually and carefully for the implications of variable elimination.

%% file: attractor-illustration.tex
\tikzset{
    circ/.pic={\fill (0,0) circle (1.2pt);}
}
\begin{tikzpicture} 
\draw[thick] (0,0) rectangle (5.2,3);
\coordinate (a1) at (0,1);
\coordinate (a2) at (2,0.2);
\coordinate (a3) at (4.2,1.5);
\coordinate (a4) at (4.0,1.95);
\coordinate (a5) at (3,2.6);
\coordinate (a6) at (2.4,2.95);
\coordinate (a7) at (1,2);
\draw[use Hobby shortcut,thick] ([out angle=-90]a1)..(a2)..([in angle=-90]a3); 
\draw[use Hobby shortcut,thick] ([out angle=90]a3)..(a4)..(a5)..(a6)..(a7)..([in angle=90]a1);

\coordinate (f) at (2.5,3.5);
\node at (f) {$f\colon\mathbb{B}^n\to\mathbb{B}^n$};

\coordinate (A) at (3.8,2.6);
\node at (A) {$\mathcal{A}$};

\coordinate (sx1) at (4.8,0.7);
\coordinate (sx2) at (3.8,0.7);
\draw (sx1) pic {circ};
\draw (sx2) pic {circ};
\node[xshift=-0.4cm,yshift=0.1cm] at (sx2) {\scriptsize $\mathcal{S}(x)$};
\draw[->,>=stealth,thick] (sx1) -- (sx2);

\draw[->,thick] (6,1.5) -- (8,1.5);
\node at (7,2.3) {\scriptsize variable};
\node at (7,2) {\scriptsize elimination};

\draw[thick] (8.5,0) rectangle (8.5 + 2.8,3);
\coordinate (ft) at (8.5 + 1.4,3.5);
\node at (ft) {$\tilde{f}\colon\mathbb{B}^{n-1}\to\mathbb{B}^{n-1}$};

\coordinate (b1) at (8.5+0,1/2);
\coordinate (b2) at (8.5+2/2,0.2/2);
\coordinate (b3) at (8.5+4.8/2,1.5/2);
\coordinate (b4) at (8.5+4.4/2,1.95/2);
\coordinate (b5) at (8.5+3.3/2,2.6/2);
\coordinate (b6) at (8.5+2.4/2,2.95/2);
\coordinate (b7) at (8.5+1/2,2/2);

\draw[use Hobby shortcut,thick] ([out angle=-90]b1)..(b2)..([in angle=-90]b3); 
\draw[use Hobby shortcut,thick] ([out angle=90]b3)..(b4)..(b5)..(b6)..(b7)..([in angle=90]b1);

\coordinate (c1) at (8.5+0,1.5+1/2);
\coordinate (c2) at (8.5+2/2,1.5+0.2/2);
\coordinate (c3) at (8.5+4.8/2,1.5+1.5/2);
\coordinate (c4) at (8.5+4.4/2,1.5+1.95/2);
\coordinate (c5) at (8.5+3.3/2,1.5+2.6/2);
\coordinate (c6) at (8.5+2.4/2,1.5+2.95/2);
\coordinate (c7) at (8.5+1/2,1.5+2/2);

\draw[use Hobby shortcut,thick] ([out angle=-90]c1)..(c2)..([in angle=-90]c3); 
\draw[use Hobby shortcut,thick] ([out angle=90]c3)..(c4)..(c5)..(c6)..(c7)..([in angle=90]c1);

\coordinate (sx3) at (8.5+4.8/2,0.6);
\draw (sx3) pic {circ};
\node[xshift=+0.2cm,yshift=0.1cm] at (sx3) {\scriptsize $x$};

\draw [->,dashed] (sx3) to [bend right=-12] (sx1);
\draw [->,dashed] (sx3) to [bend right=-14] (sx2);
\end{tikzpicture} 